\documentclass[11pt]{article}
\usepackage[margin=1in]{geometry}
\usepackage{palatino}
\usepackage{amsthm,amsmath,amssymb,amsfonts}
\usepackage[final]{graphicx}
\usepackage{subcaption}
\usepackage{hyperref}
\usepackage{booktabs}
\usepackage{algorithm, algpseudocode}
\usepackage{authblk}
\usepackage{xcolor}
\usepackage{ifdraft}
\usepackage[utf8]{inputenc}
\usepackage{xspace}
\usepackage[final]{pdfpages}
\usepackage{graphicx}
\usepackage{verbatim}
\usepackage{xfrac}
\usepackage{array}
\usepackage[normalem]{ulem}
\usepackage[labelfont=bf]{caption}
\usepackage{float}

\newcommand{\dhcre}{Detailed DHC-A\xspace}
\newcommand{\safetabp}{SafeTab-P\xspace}
\usepackage{fancyhdr}
\fancypagestyle{empty}{
   \fancyhf{} 
   
}

\theoremstyle{definition}
\newtheorem{definition}{Definition}

\theoremstyle{plain}
\newtheorem{theorem}{Theorem}
\newtheorem{lemma}{Lemma}
\newtheorem{corollary}{Corollary}
\newtheorem{proposition}{Proposition}

\newcommand{\eat}[1]{}

\newcommand{\epsfrac}{\gamma}

\graphicspath{{screenshots/}{figures/}}

\title{SafeTab-P: Disclosure Avoidance for the 2020 Census Detailed Demographic and Housing Characteristics File A (Detailed DHC-A)}
\author[1]{Sam Haney}
\author[1]{Skye Berghel}
\author[1]{Bayard Carlson}
\author[2]{Ryan Cumings-Menon}
\author[1]{Luke Hartman}
\author[1]{Michael Hay}
\author[1]{Ashwin Machanavajjhala}
\author[1]{Gerome Miklau}
\author[1]{Amritha Pai}
\author[1]{Simran Rajpal}
\author[1]{David Pujol}
\author[1]{William Sexton}
\author[1]{Ruchit Shrestha}
\author[1]{Daniel Simmons-Marengo}

\affil[1]{Tumult Labs}
\affil[2]{U.S. Census Bureau}
\date{August 23, 2024}

\begin{document}

\maketitle
\begin{abstract}
This article describes the disclosure avoidance algorithm that the U.S. Census Bureau used to protect the Detailed Demographic and Housing Characteristics File A (Detailed DHC-A)  of the 2020 Census. The tabulations contain statistics (counts) of demographic characteristics of the entire population of the United States, crossed with detailed races and ethnicities at varying levels of geography. The article describes the SafeTab-P algorithm, which is based on adding noise drawn to statistics of interest from a discrete Gaussian distribution. A key innovation in SafeTab-P is the ability to adaptively choose how many statistics and at what granularity to release them, depending on the size of a population group.
We prove that the algorithm satisfies a well-studied variant of differential privacy, called zero-concentrated differential privacy (zCDP). We then describe how the algorithm was implemented on Tumult Analytics and briefly outline the parameterization and tuning of the algorithm. 
\end{abstract}

\newpage
\tableofcontents
\newpage

\section{Introduction}

Every 10 years, the U.S. Census Bureau carries out its constitutional mandate to conduct a census of the U.S. population. The 2020 Census is the latest of such efforts, which aims to completely enumerate every person living in the United States. 
The Herculean task of ``counting everyone once, only once, and in the right place'' for the 2020 Census involves 35 operations (e.g., Address Canvassing, Nonresponse Followup, Redistricting Data Program) \cite{census-operations}. Each of these operations is responsible for a number of systems handling various aspects of the entire census undertaking, ranging from data collection and processing to the dissemination of data products to the American people. The Disclosure Avoidance System (DAS) is one small component of this vast operational ecosystem. The DAS is technically a subsystem of the Decennial Response
Processing System and it belongs to the Data Products and Dissemination (DPD) operation. The purpose of the DAS is to ensure that data shared with the public does not violate the confidentiality of individual respondents to the census. The DAS executes its duties after census responses have been collected and (eventually) processed into a database known as the Census Edited File (CEF). Then, the Decennial Tabulation System and Center for Enterprise Dissemination Services and
Consumer Innovation, both of which are also systems of the DPD operation, finish preparing the data products for public consumption. For the 2020 Census, the DAS was overhauled and rebuilt from the ground up, compared to previous censuses, to modernize its privacy protection mechanisms. The modernization effort is intended to provide individuals with the best possible protection against privacy threats, known or unknown, given today's technological landscape by deploying a DAS redesigned on cutting-edge scientific research. 

DPD releases a variety of different data products for the 2020 Census, including the Census Redistricting Data (P.L. 94-171) Summary File, Demographic and Housing Characteristics File (DHC), Demographic Profile, Detailed Demographic and Housing Characteristics File A (Detailed DHC-A), Detailed DHC-B, and Supplemental DHC (S-DHC). Each of these products are distinct and present their own unique set of challenges with regard to disclosure avoidance. Hence, the DAS deploys a variety of privacy algorithms to optimize protection and accuracy across the various data products. We note that some products share similar enough challenges to utilize the same algorithm. Despite algorithmic differences, all statistical disclosure limitation techniques fit into the same overarching privacy framework known as \emph{differential privacy}. The differential privacy framework calls for the design of algorithms that satisfy mathematically provable guarantees regarding the data publication process. See Section \ref{sec:privacy-defintions} for further details. In this paper, we use "differential privacy" loosely to mean pure differential privacy or one of its several variants. In particular, one should assume, by default, that we are referring to zero-concentrated differential privacy unless otherwise specified. As with pure differential privacy, variants such as zero-concentrated differential privacy are based on mathematically rigorous privacy definitions.

The focus of this paper, SafeTab-P, is one of the several privacy algorithms deployed by the DAS. SafeTab-P is designed specifically to provide differential privacy guarantees for the production and release of the Detailed DHC-A (see Section \ref{sec:problem} for a description of the data product and privacy release problem). 

The main goals of this article are threefold:

\begin{enumerate}
    \item Describe the SafeTab-P algorithm and how it meets the requirements of the Detailed DHC-A.
    \item Prove the privacy properties of the SafeTab-P algorithm.
    \item Describe the parameters in the SafeTab-P algorithm and how they impact privacy-accuracy tradeoffs.
\end{enumerate}

On the first point, we aim to provide a technical pseudocode description of the algorithm that articulates the functional mechanics of how privacy protection is executed on the Census data to produce the Detailed DHC-A tabular summaries (see Section \ref{sec:algorithm-description}). We will also highlight salient differences between our pseudocode abstraction of SafeTab-P and the actual implementation of the algorithm \footnote{ \url{https://github.com/uscensusbureau/DAS_2020_DDHCA_Production_Code} } (see Section \ref{sec:implementation}). The SafeTab-P algorithm is open-sourced to the public and this article may, in some regard, be viewed as a helpful companion guide to the code. It provides meaningful context for those interested in understanding the implemented algorithm. However, this article is not a code specification document, per se. That is, a deep-dive into the code architecture (e.g., module interactions and class descriptions) is out of scope.

On the second point, we will provide relevant background material on the differential privacy framework and explain why SafeTab-P fits into this framework (see Section \ref{sec:safetab-discrete-gaussian-privacy}). Since this paper primarily emphasizes the privacy properties of SafeTab-P, it is not intended as a user guide to the data product. Guidance for users can be found in the Detailed DHC-A technical documentation produced by the Census Bureau \cite{race-spec}. 

Finally, as the third point indicates, we discuss the parameters in SafeTab-P that impact the privacy-accuracy tradeoff of the algorithm. We cover parameter tuning, including a brief look at related data accuracy considerations (see Section \ref{sec:params}). Parameter tuning requires correctly identifying the universe of parameters and understanding the trade-offs inherent in adjusting these parameters. 

\section{Problem \& Desiderata}
\label{sec:problem}
The Detailed DHC-A consists of statistics (counts) of demographic characteristics for all persons in the United States and Puerto Rico crossed with detailed races and ethnicities at varying levels of geography. The demographic characteristics include population counts and sex by age statistics for 2996 race and ethnicity characteristic iterations (defined in Section \ref{sec:race-defs}). Despite the name, the Detailed Demographic and \emph{Housing} Characteristic File A product does not provide any data on households. Household data for detailed race and ethnicity groups and American Indian and Alaska Native tribes and villages is available in the Detailed DHC-B. The Census Bureau separated these products ``to better facilitate developing disclosure protections for these complex data" \cite{race-spec}. In this section, we define relevant concepts, outline the precise statistics to be released, and then formulate the differentially private algorithm design problem. 

\subsection{Geography}
Every person resides in exactly one Census block that determines their geographic location. Census blocks are the most granular form of \textit{geographic entities}. All other geographic entities (e.g., L.A. county, the state of CA, and the United States) are aggregations of Census blocks. Geographic entities are divided into \emph{geographic summary levels}. A geographic summary level is a set of nonoverlapping geographic entities, such as the set of all states, or the set of all counties. Detailed DHC-A produces statistics for the following geographic summary levels:

\begin{itemize}
    \item Nation
    \item State or State equivalent
    \item County or County equivalent
    \item Census Tract
    \item Place
    \item American Indian, Alaska Native, and Native Hawaiian (AIANNH) areas.
\end{itemize}

Henceforth, we tend to write State (County) without including the ``or State (County) equivalent'' qualifier although one should assume the qualifier when applicable. Washington, D.C. is an example of a State equivalent.

\subsection{Race and Ethnicity}
\label{sec:race-defs}
The U.S. Census Bureau collects detailed race and ethnicity data from individuals in accordance with the 1997 Federal Register Notice ``Revisions to the Standards for the Classification of Federal Data on Race and Ethnicity'' released by the Office of Management and Budget (OMB). 

Every person is associated with one or more \textit{race codes} and, per OMB guidelines, a single \textit{ethnicity code}. The maximum number of race codes, called the \textit{race multiplicity}, that a person can be associated with is limited to 8 by the census data collection procedures. A \textit{race group} is a set of race codes. Similarly, an ethnicity group is a set of ethnicity codes.

\textit{Detailed} race or ethnicity groups are the most disaggregated racial or ethnic group classifications for which the Census Bureau publishes data.  Examples of Detailed racial or ethnic groups include Basque, Dutch, Guatemalan, Puerto Rican, Ethiopian, Nigerian, Mongolian, Thai, Apache, and Navajo Nation. The OMB-specified \textit{Major} race categories are aggregated race groupings that represent the minimum allowable race classification system for which census data may be published. Major race grouping include White, Black or African American, American Indian or Alaska Native, Asian, Native Hawaiian or Other Pacific Islander, and Some Other Race. The aggregated ethnic equivalent of the Major race groups is a coarse binary classification (Hispanic or Latino, Not Hispanic or Latino). \textit{Regional} race or ethnicity groups provide an intermediate level of aggregation between Detailed and Major. Examples of Regional racial or ethnic groups include European, Central American, Caribbean, Sub-Saharan African, Alaska Native, and American Indian. Note that, while Regional groups are often distinguished by geographic location (Europe, Sub-Saharan Africa, Caribbean, etc.), the Regional concept pertains to a classification of race and ethnicity; it does not pertain to geography. Subject-matter experts at the Census Bureau determine which race or ethnicity groups are Detailed and which groups are Regional. For the purposes of SafeTab-P, we take these classifications as given exogenous factors. The universe specification of valid race and ethnicity groups as well as the classification into Detailed or Regional groups began well before data collection for the 2020 Census \cite{race-blog}. Appendix G of \cite{race-spec} provides a complete enumeration of the detailed race and ethnicity groups. 

An individual person is in a race group \textit{Alone} if all race codes associated with that individual are contained in the race group. For example, if an individual self-identifies with the single race code for Navajo Nation and no other race codes, said individual would belong to the Detailed race group Navajo Nation Alone. Alternatively, an individual may report multiple race codes (e.g., British, Scottish, and Dutch) that aggregate into the same Regional group (e.g., European Alone). A person is in a race group \textit{Alone or in Any Combination} if some race code associated with that person is contained in the race group. This concept pertains to individuals that self-identify with a single race (e.g., British) or with multiple races (e.g., British and Thai). In both examples, the individual belongs to the Detailed British race group Alone or in Any Combination. The individual also belongs to the Regional European race group Alone or in Any Combination. Since all individuals are only associated with a single ethnicity code, all ethnicity groups are Alone groups.

A \textit{characteristic iteration} is the combination of a race (or ethnicity) group, along with the specification of either Alone or Alone or in Any Combination (e.g., Latin American Indian Alone or in Any Combination is a characteristic iteration). One person may be associated with multiple characteristic iterations. Like geographical entities, characteristic iterations are also divided into \emph{characteristic iteration levels}. We have already seen the defining aspect of these iteration levels: namely, the concepts of \emph{Detailed} and \emph{Regional} race groups. The Detailed characteristic iteration level consists of the set of characteristic iterations for all Detailed race and ethnicity groups Alone as well as for all Detailed race groups Alone or in Any Combination (e.g., Chinese Alone, Chinese Alone or in Any Combination, Celtic Alone, and Celtic Alone or in Any Combination). The Regional characteristic iteration level consists of the set of characteristic iterations for all Regional race and ethnicity groups Alone as well as for all Regional race groups Alone or in Any Combination (e.g., Middle Eastern or North African Alone, Middle Eastern or North African Alone or in Any Combination, Polynesian Alone, and Polynesian Alone or in Any Combination). We intentionally omit the notion of a Major race and ethnicity characteristic iteration level, as no statistics for this level will be produced by the SafeTab-P algorithm for the Detailed DHC-A.

\subsection{Population Groups}\label{sec:population-group-levels}

A \textit{population group} is a pair $(g, c)$, where $g$ is a geographic entity (e.g., the state of NC or L.A. County) and $c$ is a race or ethnicity characteristic iteration (e.g., Latin American Indian Alone or in Any Combination). Population groups are divided into \textit{population group levels}. We will often identify a population group level by specifying a (geography level, characteristic iteration level) pair. However, each population group level is really a set of population groups, where each population group's geographic entity belongs to the specified geography level and its characteristic iteration belongs to the specified characteristic iteration level. More formally, the Detailed DHC-A requires the publication of statistics for the following population group levels:

\begin{itemize}
    \item (Nation, Detailed) $\equiv$ $\{(g, c): g \text{ is the Nation}, c \text{ is a Detailed characteristic iteration}\}$
    \item (State, Detailed) $\equiv$ $\{(g, c): g \text{ is a State}, c \text{ is a Detailed characteristic iteration}\}$
    \item (County, Detailed) $\equiv$ $\{(g, c): g \text{ is a County)}, c \text{ is a Detailed characteristic iteration}\}$
    \item (Tract, Detailed) $\equiv$ $\{(g, c): g \text{ is a Census Tract}, c \text{ is a Detailed characteristic iteration}\}$
    \item (Place, Detailed) $\equiv$ $\{(g, c): g \text{ is a Place}, c \text{ is a Detailed characteristic iteration}\}$
    \item (AIANNH, Detailed) $\equiv$ $\{(g, c): g \text{ is an AIANNH area}, c \text{ is a Detailed characteristic iteration}\}$
    \item (Nation, Regional) $\equiv$ $\{(g, c): g \text{ is the Nation}, c \text{ is a Regional characteristic iteration}\}$
    \item (State, Regional) $\equiv$ $\{(g, c): g \text{ is a State}, c \text{ is a Regional characteristic iteration}\}$
    \item (County, Regional) $\equiv$ $\{(g, c): g \text{ is a County}, c \text{ is a Regional characteristic iteration}\}$
    \item (Tract, Regional) $\equiv$ $\{(g, c): g \text{ is a Census Tract}, c \text{ is a Regional characteristic iteration}\}$
    \item (Place, Regional) $\equiv$ $\{(g, c): g \text{ is a Place}, c \text{ is a Regional characteristic iteration}\}$
\end{itemize}

In practice, some Detailed or Regional characteristic iterations may be omitted from the tabulations in a geography level. In other words, the above population group levels are subsets of the designated sets. This is done in accordance with specifications for the Detailed DHC-A provided by the Census Bureau. There is no concise representation for the exact level sets. The population group level (AIANNH, Regional) is intentionally omitted from the Detailed DHC-A.

One person may belong to multiple population groups in the set that comprises a population group level. For example, an individual who resides in Texas and reports Kenyan and Ghanaian races would belong to the (Texas, Kenyan Alone or in Any Combination) and the (Texas, Ghanaian Alone or in Any Combination) population groups which are both contained in the population group level identified by (State, Detailed). A person who resides in Schuyler County, NY and only reports a single race of Dutch would still belong to the (Schuyler County, NY, Dutch Alone) and the (Schuyler County, NY, Dutch Alone or in Any Combination) population groups which are both contained in the level identified by (County, Detailed). One person may be connected with population groups across multiple population group levels. For example, the Dutch individual residing in Schuyler County, NY would additionally belong to the (NY, Dutch Alone) and (NY, Dutch Alone or in Any Combination) population groups in the (State, Detailed) level, the (Schuyler County, NY, European Alone) and (Schuyler County, NY, European Alone or in Any Combination) population groups in the (County, Regional) level, etc. It is possible for an individual to not belong to any population groups in a particular level. Specifically, because AIANNH areas do not cover the United States, an individual who resides outside all designated AIANNH areas does not belong to any population groups in the (AIANNH, Detailed) level.  

Because the geographic entities in a geography level are disjoint (a person cannot reside in both Schuyler County, NY and Fairfax County, VA), an individual's characteristic iterations primarily determine the number of population groups the individual belongs to in each level. A person associated with the maximum of 8 race codes and a single ethnicity code could belong to at most 9 Detailed characteristic iterations (8 Alone or in Any Combination race groups and 1 Alone ethnic group) and, similarly, at most 9 Regional characteristic iterations. Thus, for any given population group level, the maximum number of population groups an individual may contribute to is 9. 

\subsection{Detailed Demographic and Housing Characteristics File A}
We are now prepared to define the \dhcre. The product aims to tabulate statistics by population groups.  

The following statistical tables are released for each eligible population group as part of Detailed DHC-A.\footnote{Population groups may be deemed ineligible to receive certain statistics for a variety of reasons discussed throughout the document, such as groups pre-specified to only receive T01001 counts and groups that may fail to meet certain population thresholds).}   
\begin{itemize}
    \item A total population table (T01001) associated with each population group. See Table \ref{tab:t1-shell}.
    \item Sex by age counts for a subset of population groups. The sex by age tables come in three different variants (T02001, T02002, T02003) for reasons explained later. See Tables \ref{tab:t2-age4-shell}, \ref{tab:t2-age9-shell}, \ref{tab:t2-age23-shell}.
\end{itemize}

\begin{table}[H]
\small
\begin{tabular}{l}
\textbf{T01001 Total Population}\\
\textit{Universe: Total Population}\\
Total\\
\end{tabular}
\caption{\label{tab:t1-shell} The T01001 table contains counts of persons for all eligible population groups.}
\end{table}

\begin{table}[H]
\small
\begin{tabular}{l}
\textbf{T02001 Sex by Age(4)}\\
\textit{Universe: Total Population}\\
Total\\
\quad \quad Male \\
\quad \quad \quad \quad Under 18 years \\
\quad \quad \quad \quad 18 to 44 years \\
\quad \quad \quad \quad 45 to 64 years \\
\quad \quad \quad \quad 65 years and over \\
\quad \quad Female \\
\quad \quad \quad \quad Under 18 years \\
\quad \quad \quad \quad 18 to 44 years \\
\quad \quad \quad \quad 45 to 64 years \\
\quad \quad \quad \quad 65 years and over \\
\end{tabular}
\caption{\label{tab:t2-age4-shell} The T02001 Sex by Age (4) table contains counts of persons by sex and age in 4 bins for each eligible population group.}
\end{table}

\begin{table}[H]
\small
\begin{tabular}{l}
\textbf{T02002 Sex by Age(9)}\\
\textit{Universe: Total Population}\\
Total\\
\quad \quad Male \\
\quad \quad \quad \quad Under 5 years \\
\quad \quad \quad \quad 5 to 17 years \\
\quad \quad \quad \quad 18 to 24 years \\
\quad \quad \quad \quad 25 to 34 years \\
\quad \quad \quad \quad 35 to 44 years \\
\quad \quad \quad \quad 45 to 54 years \\
\quad \quad \quad \quad 55 to 64 years \\
\quad \quad \quad \quad 65 to 74 years \\
\quad \quad \quad \quad 75 years and over \\
\quad \quad Female \\
\quad \quad \quad \quad Under 5 years \\
\quad \quad \quad \quad 5 to 17 years \\
\quad \quad \quad \quad 18 to 24 years \\
\quad \quad \quad \quad 25 to 34 years \\
\quad \quad \quad \quad 35 to 44 years \\
\quad \quad \quad \quad 45 to 54 years \\
\quad \quad \quad \quad 55 to 64 years \\
\quad \quad \quad \quad 65 to 74 years \\
\quad \quad \quad \quad 75 years and over \\
\end{tabular}
\caption{\label{tab:t2-age9-shell} The T02002 Sex by Age (9) table contains counts of persons by sex and age in 9 bins for each eligible population group.}
\end{table}

\begin{table}[H]
\small
\begin{tabular}{l}
\textbf{T02003 Sex by Age(23)}\\
\textit{Universe: Total Population}\\
Total\\
\quad \quad Male \\
\quad \quad \quad \quad Under 5 years \\
\quad \quad \quad \quad 5 to 9 years \\
\quad \quad \quad \quad 10 to 14 years \\
\quad \quad \quad \quad 15 to 17 years \\
\quad \quad \quad \quad 18 and 19 years \\
\quad \quad \quad \quad 20 years \\
\quad \quad \quad \quad 21 years \\
\quad \quad \quad \quad 22 to 24 years \\
\quad \quad \quad \quad 25 to 29 years \\
\quad \quad \quad \quad 30 to 34 years \\
\quad \quad \quad \quad 35 to 39 years \\
\quad \quad \quad \quad 40 to 44 years \\
\quad \quad \quad \quad 45 to 49 years \\
\quad \quad \quad \quad 50 to 54 years \\
\quad \quad \quad \quad 55 to 59 years \\
\quad \quad \quad \quad 60 to 61 years \\
\quad \quad \quad \quad 62 to 64 years \\
\quad \quad \quad \quad 65 and 66 years \\
\quad \quad \quad \quad 67 to 69 years \\
\quad \quad \quad \quad 70 to 74 years \\
\quad \quad \quad \quad 75 to 79 years \\
\quad \quad \quad \quad 80 to 84 years \\
\quad \quad \quad \quad 85 years and over \\
\quad \quad Female \\
\quad \quad \quad \quad Under 5 years \\
\quad \quad \quad \quad 5 to 9 years \\
\quad \quad \quad \quad 10 to 14 years \\
\quad \quad \quad \quad 15 to 17 years \\
\quad \quad \quad \quad 18 and 19 years \\
\quad \quad \quad \quad 20 years \\
\quad \quad \quad \quad 21 years \\
\quad \quad \quad \quad 22 to 24 years \\
\quad \quad \quad \quad 25 to 29 years \\
\quad \quad \quad \quad 30 to 34 years \\
\quad \quad \quad \quad 35 to 39 years \\
\quad \quad \quad \quad 40 to 44 years \\
\quad \quad \quad \quad 45 to 49 years \\
\quad \quad \quad \quad 50 to 54 years \\
\quad \quad \quad \quad 55 to 59 years \\
\quad \quad \quad \quad 60 to 61 years \\
\quad \quad \quad \quad 62 to 64 years \\
\quad \quad \quad \quad 65 and 66 years \\
\quad \quad \quad \quad 67 to 69 years \\
\quad \quad \quad \quad 70 to 74 years \\
\quad \quad \quad \quad 75 to 79 years \\
\quad \quad \quad \quad 80 to 84 years \\
\quad \quad \quad \quad 85 years and over \\
\end{tabular}
\caption{\label{tab:t2-age23-shell} The T02003 Sex by Age(23) table contains counts of persons by sex and age in 23 bins for each eligible population group.}
\end{table}

Some population groups are pre-determined to only receive T01001 statistics. These groups will be called \emph{TotalOnly} population groups throughout the document. The TotalOnly population groups are only tabulated at the Nation and State geography levels. That is, TotalOnly is a subset of the set union of (Nation, Detailed), (Nation, Regional), (State, Detailed), (State, Regional). The categorization of population groups as TotalOnly is exogenous to the SafeTab-P program. As with other population group categorizations, this determination is made by subject-matter experts outside the DAS. The Detailed DHC-A technical documentation explains, ``detailed groups with a national population less than 50 in the 2010 Census were preset to only receive nation and state level total population counts'' \cite{race-spec}. Importantly, the classification was not based on collected 2020 Census responses and thus, does not impact the confidentiality protections afforded by the SafeTab-P program.

\subsection{Private Release Problem}\label{sec:privacy-problem}
The release of statistical data products by the U.S. Census Bureau about persons and households is regulated under Title 13 and any release of statistics about persons in the United States must be afforded strong privacy protections \cite{title13}. Moreover, it has been demonstrated that legacy statistical disclosure limitation (SDL) techniques are vulnerable to attacks that can reconstruct the sensitive person records from aggregate statistics \cite{ap-census-attack}. Hence, the U.S. Census Bureau decided to release many of the 2020 Census data products, including \dhcre, using algorithms that satisfy modern privacy definitions like differential privacy \cite{dsep-dp}.

In particular, we describe a disclosure avoidance technique for the Detailed DHC-A that was designed to satisfy the following desiderata:  

\begin{itemize}
    \item \textit{Privacy:} The algorithm must ensure end-to-end zero-concentrated differential privacy with respect to (the addition/removal of) every person in the United States and Puerto Rico.
    \item \textit{Population Groups:} The algorithm must release statistics for a predefined set of race and ethnicity characteristic iterations and the following geographies: nation, states, counties, Census tracts, places, and all areas designated as AIANNH areas.  
    \item \textit{Adaptivity:} The algorithm may adaptively choose the granularity at which sex by age statistics are released. For instance, for population groups with a few people, the sex by age histogram may only have 4 buckets of age (as seen in Table \ref{tab:t2-age4-shell}, while for population groups with many people a more detailed histogram may be released (such as the 23 age bins shown in Table \ref{tab:t2-age23-shell}). 
    \item \textit{Accuracy:} The algorithm must achieve pre-specified accuracy levels for population groups in terms of the margin of error (MOE) in output counts. Different population groups may have different MOEs specified (described later in the paper in Table~\ref{tab:moe-targets}). The MOE discussed in the paper captures error induced by disclosure avoidance induced and does not capture other sources of error such as under/over counting in the census. 
    \item \textit{Integrality:} The output statistics must be integers. 
    \item \textit{Minimal Consistency:} The algorithm is not required, in general, to ensure consistency. That is, different counts output by the system need not be consistent with each other (e.g., the number of people of a certain characteristic iteration in the United States need not equal the sum of the population counts for the same characteristic iteration across all states). We also note that no consistent is enforced with other data products such as the DHC. However, some postprocessing of outputs is done to address specific demographic reasonableness concerns. These postprocessing steps are discussed in Section \ref{sec:postprocessing}.
\end{itemize}
In the rest of the paper, we describe the SafeTab-P differential privacy algorithm, discuss implementation and parameter tuning, and analyze bounds on the privacy loss achievable while satisfying the constraints mentioned above.   

\section{SafeTab-P Algorithm}
\label{sec:algorithm-description}
\safetabp is a privacy algorithm for releasing detailed race and ethnicity statistics from the 2020 Census. The algorithm must accommodate the release of tabulations for total counts by detailed race and ethnicity and tabulations for sex by age counts by detailed race and ethnicity for various geographic summary levels. The algorithm acts on a private dataframe derived from the 2020 Census. The algorithm does not re-use noisy estimates or privacy-loss budget from other 2020 data products such as the DHC. In this section, we will describe the algorithm as applied to the United States. Puerto Rico is discussed in Section \ref{sec:pr}.

\subsection{Input Data Description}
The 2020 CEF is a relational database consisting of multiple person and household attributes spread across several linked dataframes. Many of these attributes are irrelevant to the tabulations in the Detailed DHC-A. As such, we assume a simplified, reduced-form data representation that is sufficient for our purposes. That is, we imagine deriving a single private dataframe from the 2020 CEF that consists of a row for each person in the United States with the following attributes: BlockID, RaceEth, Sex, Age.

\vspace{\baselineskip}

\noindent \textbf{BlockID} is a single attribute that geolocates a person record to a unique Census block. As previously discussed, all geographic entities are aggregations of blocks. Thus, we assume that BlockID implicitly encodes each record's unique tract, county, and state. BlockID also encodes whether a record belongs to an AIANNH area and, if so, uniquely identifies the area. We note that all records are vacuously included in the Nation geographic entity.

\vspace{\baselineskip}

\noindent \textbf{RaceEth} is a single attribute that encodes up to eight race codes and an ethnicity code. That is, one person's RaceEth attribute may indicate the person is Andorran and Dominican while another person's RaceEth attribute indicates Chinese, Japanese, Tongan, and Not Hispanic or Latino. We assume this RaceEth conceptualization combined with Census Bureau specifications fully determines the characteristic iterations of an individual.

\vspace{\baselineskip}

\noindent \textbf{Sex} is recorded as Male or Female in the 2020 Census.

\vspace{\baselineskip}

\noindent \textbf{Age} is recorded in single-year increments (e.g., 12, 30, 82). 

\vspace{\baselineskip}

\subsection{The Algorithm Description}

SafeTab-P must produce tabulations for population groups. As previously defined, a population group is a geographic entity (e.g. a specific county) and a characteristic iteration code (see Section~\ref{sec:problem} for more details). Population groups are split into sets called population group levels (specified by a geography level and an iteration level) with distinct privacy-loss budgets.  Records are associated with the population groups of a population group level via transformations that map their BlockID to geographic entities and their RaceEth attribute to characteristic iterations. In our simplified model, we assume there exists a master list of all population groups divided into population group levels. This master list may include \emph{empty} population groups (combinations of geographic entities and characteristic iterations with no associated records in the private dataframe).  We assume the following model for population groups:

\begin{itemize}
    \item SafeTab-P produces tabulations for population group levels $\mathcal{P}_1,\ldots,\mathcal{P}_{\omega}$. That is, it should produce a tabulation for each population group $P \in \mathcal{P}_i$ for $1 \le i \le {\omega}$. For example, $\mathcal{P}_i$ may be the level (State, Detailed) consisting of population groups, such as (Iowa, Albanian Alone) and (Kansas, German Alone or in any Combination).
    \item SafeTab-P is provided  privacy-loss budgets for each population group level $\rho_1, \ldots, \rho_{\omega}$ with $\rho_i$ corresponding to the budget for population group level $\mathcal{P}_i$. Privacy-loss budgets are described in greater detail in Section \ref{sec:privacy-prelim}. For now, we note that each $\rho_i$ is a positive real-valued number.
    \item For each $\mathcal{P}_i$, we assume we have a function $g_i: \mathcal{I} \rightarrow 2^{\mathcal{P}_i}$, where $\mathcal{I}$ is the domain of records in the private dataframe. That is, $g_i$ maps a record $r$ to the subset of population groups at level $i$ to which it belongs (i.e., $g_i(r) \subset \mathcal{P}_i$). For example, suppose $i$ corresponds to the (State, Regional) level, the record $r$'s BlockID uniquely identifies the state of residence as Idaho and its RaceEth attribute encodes Nigerian, Beninese, Tongan, and Not Hispanic or Latino. Then $g_i(r)$ would associate the record with the following population groups: (Idaho, Sub-Saharan African, Alone or in any Combination) and (Idaho, Polynesian Alone or in any Combination).  
    \item We assume the stability of $g_i$, denoted by $\Delta(g_i)$, is known. The stability is defined as the maximum number of population groups a record could belong to in a level. Formally, $\Delta(g_i) = \max_{r \in \mathcal{I}} |g_i(r)|$ \cite{McSherry09}. Importantly, this value defines what could be the maximum based on any hypothetically possible record, rather than defining what is the maximum based on the collected 2020 Census data. In other words, the stability is a data-independent value. As described earlier, this value is $\Delta(g_i) = 9$ for all population group levels tabulated in Detailed DHC-A. 
\end{itemize}

The main algorithm is presented in Algorithm~\ref{alg:safetab-main-algorithm}.
This algorithm proceeds by looping over the population group levels.
For each population group level, we apply $g_i$ to the dataframe to map each record to the set of population groups it is associated with.
Then, for each population group in the level, we call the tabulation function \textsc{TabulatePopulationGroup}, passing in a dataframe containing just the records in that population group.

The pseudocode for the procedure \textsc{TabulatePopulationGroup} is given in Algorithm~\ref{alg:safetab-tabulate-pop-group}. This code tabulates a single population group.
Population groups are characterized based on the tabulation we would like to compute. In particular, we assume we are given a set TotalOnly of population groups for which only a T01001 total count of the population group should be tabulated. We check whether the given group is a member of this set.
If it is, we call the \textsc{NoisyCount} function on the population group, which tabulates a noisy total count for the group.
Otherwise, we use a two stage algorithm. The computations in both stages require noise infusion to ensure the entire algorithmic process is covered by the differential privacy guarantee.
We first compute a noisy count of the group using \textsc{NoisyCount} but using only a fraction (denoted $\epsfrac$) of the available privacy-loss budget. Next, we compare this noisy count against a set of given thresholds, denoted $\Theta_1, \Theta_2,$ and $\Theta_3$.  Depending on which thresholds the noisy count exceeds, we compute sex by age noisy counts with a varying degree of age bin sizes. Age bins are coarser for smaller noisy counts. These sex by age counts are also computed by \textsc{NoisyCount} using the remaining privacy-loss budget. 

The pseudocode for the procedure \textsc{NoisyCount} is given in Algorithm~\ref{alg:base-discrete-gaussian}. This procedure computes the number of rows in the dataframe and adds noise from a discrete Gaussian distribution. 

The notation used in this section and the algorithm pseudocode is summarized in Table~\ref{tab:algorithm-notation}.

\begin{table}[t]
    \centering
    \begin{tabular}{c p{.8\linewidth}}
        \toprule
        \textbf{Notation} & \textbf{Description} \\
        \midrule
        $\omega$ & The number of population group levels \\
        $\mathcal{P}_i$ & Population group level $i$  \\
        $\rho_i$ & The privacy-loss budget allocated to population group level $i$ \\
        $g_i$ & A function mapping records to the set of population groups in $\mathcal{P}_i$ to which the record belongs \\
        $\Delta(g_i)$ & $\max_{r \in \mathcal{I}} |g_i(r)|$
    \end{tabular}
    \caption{A summary of the notation used in Section~\ref{sec:algorithm-description}}
    \label{tab:algorithm-notation}
\end{table}

\begin{algorithm}[t]
\caption{\label{alg:safetab-main-algorithm} The main \safetabp algorithm.}
\begin{algorithmic}[1]
\Require
$df$: A private dataframe with attributes [BlockID, RaceEth, Sex, Age] and one row for each person in the United States.
\Require $\{\rho_i\}_{i \in [1,\omega]}$: Privacy loss parameters for each population group level $i \in [1, \omega]$.
\Require $\gamma$: The fraction of the privacy-loss budget to be used in Stage 1 of the two stage tabulation algorithm.

\Procedure{SafeTab-P}{$df$, $\{\rho_i\}$, $\gamma$}

\For{$i \in [1, \omega]$} \label{line:pop-group-level-loop}
\State $df_i \leftarrow df$.flatmap($g_i$);
\Comment{$df_i$ has schema [PopGroup, Sex, Age]}
\State $s \leftarrow \Delta(g_i)$ 
\Comment{1 row in $df$ may result in $\leq s$ rows in $df_i$}
\For{$P \in \mathcal{P}_i$ \label{line:iteration-loop}}
\State $df_P$ $\leftarrow$ $df_i$.filter(PopGroup $== P$)
\State \Call{TabulatePopulationGroup}{$df_p$, $P$, $\rho_i/s$, $\epsfrac$}
\EndFor
\EndFor
\EndProcedure
\end{algorithmic}
\end{algorithm}

\begin{algorithm}[t]
\caption{\label{alg:safetab-tabulate-pop-group} Subroutine of \safetabp to tabulate statistics for a single population group.}
\begin{algorithmic}[1]
\Require $df$: A private dataframe with attributes [PopGroup, Sex, Age]. This dataframe should contain the records in the population group.
\Require $P$: The population group.
\Require $\rho$: Privacy-loss budget for this subroutine.
\Require $\epsfrac$: Fraction of $\rho$ used in the adaptive algorithm

\Procedure{TabulatePopulationGroup}{$df, P, \rho, \epsfrac$}

\If{$P$ $\in$ TotalOnly }
\State \textit{// For TotalOnly population groups, only report noisy total counts}
\State \textbf{Output} \Call{NoisyCount}{df.count(), $\rho$} \label{line:noisy-total-only}

\State
\Else
\State \textit{// For the rest of the population groups, adaptively choose the statistics released}
\State \textit{//  based on the noisy total count of the population group.}
\State \textit{// Step 1: Compute the noisy total count using $\epsfrac\rho$ privacy-loss budget}
\State total $\leftarrow$ \Call{NoisyCount}{df.count(), $\epsfrac\rho$} \label{line:stage1-count}
\Comment{Compute noisy total}
\State
\State \textit{// Step 2: Release statistics based on the noisy count with $(1 - \epsfrac)\rho$ privacy-loss budget}

\If{total $< \Theta_1$}
\State \textbf{Output} \Call{NoisyCount}{df.count(), $(1-\epsfrac)\rho$}\label{line:age1} 
\Comment{Output the total}

\ElsIf{total $< \Theta_2$}
df\_group $\leftarrow$ df.map(Age $\rightarrow$ Age4).groupby(Sex, Age4)
\State \textbf{Output} {\Call{NoisyCount}{df\_group.count(), $(1-\epsfrac)\rho$}} \label{line:age4}
\Comment{Sex X Age4 marginal}

\ElsIf{total $< \Theta_3$}
df\_group $\leftarrow$ df.map(Age $\rightarrow$ Age9).groupby(Sex, Age9)
\State \textbf{Output} {\Call{NoisyCount}{df\_group.count(), $(1-\epsfrac)\rho$}} \label{line:age9}
\Comment{Sex X Age9 marginal}

\Else 
df\_group $\leftarrow$ df.map(Age $\rightarrow$ Age23).groupby(Sex, Age23)
\State \textbf{Output} {\Call{NoisyCount}{df\_group.count(), $(1-\epsfrac)\rho$}}\label{line:age23}
\Comment{Sex X Age23 marginal}

\EndIf
\EndIf
\EndProcedure
\end{algorithmic}
\end{algorithm}

\begin{algorithm}[t]
\caption{\label{alg:base-discrete-gaussian} The discrete Gaussian mechanism for vectors.}
\begin{algorithmic}[1]
\Require $a$: An $n$ dimensional vector of integers.
\Require $\rho$: A privacy-loss parameter.
\Procedure{NoisyCount}{$a, \rho$}
\State $y \gets \mathcal{N}^{n}_{\mathbb{Z}}\left(\frac{1}{2\rho}\right)$
\State \textbf{return} $a + y$
\EndProcedure
\end{algorithmic}
\end{algorithm}

\clearpage

%%% Local Variables:
%%% mode: latex
%%% TeX-master: "main"
%%% End:

\section{SafeTab-P Privacy Analysis}
\label{sec:safetab-discrete-gaussian-privacy}

The algorithms employed by the Census Bureau's DAS to preserve the privacy of individual census respondents must adhere to the privacy standards imposed by zero-concentrated differential privacy (zCDP). In this section, we show that the SafeTab-P algorithm presented in Section~\ref{sec:algorithm-description} does satisfy the requirements of zCDP. Before presenting the proof, we first provide necessary background on zCDP. 

\subsection{Privacy Preliminaries}
\label{sec:privacy-prelim}

We provide the formal mathematical definition of zCDP and then describe relevant privacy properties that zCDP ensures. 

\subsubsection{Privacy definitions}
\label{sec:privacy-defintions}

\begin{definition}[Neighboring Databases]
\label{def:neighboring-databases}
Let $x,x'$ be databases represented as multisets of tuples. We say that $x$ and $x'$ are \emph{neighbors} if their symmetric difference is 1.
\end{definition}

We define zCDP, which bounds the \emph{R\'enyi divergence} between the distributions of a mechanism run on neighboring databases.

\begin{definition}
The \emph{R\'enyi divergence of order $\alpha$} between distribution $P$ and distribution $Q$, denoted $D_{\alpha}(P \| Q)$, is defined as
\begin{equation}
    D_{\alpha}(P \| Q) = \frac{1}{\alpha-1}\log\left(\underset{x \sim P}{\mathbb{E}} \left[ \left( \frac{P(x)}{Q(x)} \right)^{\alpha-1}\right]\right)
\end{equation}
\end{definition}

\begin{definition}(zCDP \cite{BunS16})\label{def:zcdp}
An algorithm $M: \mathcal{X} \rightarrow \mathcal{Y}$ satisfies $\rho$-zCDP if for all neighboring $x, x' \in \mathcal{X}$ and for all $\alpha \in (1, \infty)$,
\begin{equation}
    D_{\alpha}(M(x) \| M(x')) \le \alpha\rho.
\end{equation}
\end{definition}

The value $\rho$ is often called the privacy-loss budget of an algorithm. Definition \ref{def:zcdp} is sometimes called ``unbounded'' $\rho$-zCDP. There is an alternative definition known as ``bounded'' $\rho$-zCDP. The difference between (unbounded) zCDP (Definition~\ref{def:zcdp}) and bounded zCDP (to be defined in Definition~\ref{def:bounded-zcdp}) is the definition of neighboring databases.
Informally, neighboring databases under unbounded zCDP are obtained by adding/removing a record, while neighbors under bounded zCDP are obtained by changing one record\footnote{The word ``bounded'' comes from the fact that all neighboring databases are the same size.}.
This difference in definitions results in a difference in the semantics of the privacy guarantee.
Briefly, unbounded zCDP protects the presence of any record in the database, while bounded zCDP protects the value of each record in the database.
One notable difference is that bounded zCDP does not protect the total number of records in the database.

The formal definition of bounded neighbors is the following:

\begin{definition}[Bounded Neighboring Databases]
  \label{def:bounded-neighboring-databases}
  Let $x,x'$ be databases represented as multisets of tuples.
  We say that $x$ and $x'$ are \emph{bounded neighbors} if they contain the same number of tuples and are identical except for a single tuple.
\end{definition}

Using this new definition of neighbors, we can give the definition of bounded zCDP.

\begin{definition}[Bounded zCDP]
\label{def:bounded-zcdp}
An algorithm $M: \mathcal{X} \rightarrow \mathcal{Y}$ satisfies bounded $\rho$-zCDP if for all bounded neighboring databases $x, x' \in \mathcal{X}$, for all output $y \in \mathcal{Y}$, and for all $\alpha \in (1, \infty)$,
\begin{equation}
    D_{\alpha}(M(x) \| M(x')) \le  \alpha\rho.
\end{equation}
\end{definition}

The Tumult-developed libraries (see Section \ref{sec:analytics}) used to implement SafeTab-P provide unbounded zCDP guarantees. We provide an additional proof for bounded zCDP to allow for consistent comparison across other 2020 Census data releases such as the DHC that utilize bounded zCDP. A common neighbors definition allows for composition analysis across different releases rather than treating each release in isolation. The Census Bureau interest in bounded neighbors stems from a desire to model attackers that know the true database size given that several 2020 Census data products release the exact total U.S. population (according to the CEF) without noise infusion.

Unless otherwise specified, we assume unbounded $\rho$-zCDP, by default, in this paper.

\subsubsection{Composition Theorems}
\label{sec:composition}

One of the most useful and important properties of privacy definitions is their behavior under composition.
In this section, we state composition results for zCDP.

There are two types of composition we are interested in -- sequential composition and parallel composition. We first state the sequential composition results. 

\begin{lemma}(Adaptive sequential composition of zCDP \cite{BunS16})
\label{lem:sequential-composition-zcdp}
Let $M_1: \mathcal{X} \rightarrow \mathcal{Y}$ and $M_2: \mathcal{X} \times \mathcal{Y} \rightarrow \mathcal{Z}$ be mechanisms satisfying $\rho_1$-zCDP and $\rho_2$-zCDP respectively. Let $M_3(x) = M_2(x, M_1(x))$. Then $M_3$ satisfies $(\rho_1 + \rho_2)$-zCDP.
\end{lemma}

Next, we state and prove a generalized parallel composition lemma for our privacy definition. To the best of our knowledge, these results are novel. The standard statement of parallel composition is a special case of our generalization.

Let the \emph{maximum degree} of a set family $F = \{S_i\}$, $S_i \subseteq S$ be the maximum number of sets containing any fixed element of $S$. That is, 
\begin{equation}
    degree(F) = max_{s \in S} |\{S_i \in F | s \in S_i\}|
\end{equation}

\begin{lemma}
\label{lem:generalized-parallel-composition-zcdp}
Let $F = \{S_1,...,S_k\}$ be a family of subsets of the input domain with maximum degree $z$. Let $M_1,\ldots,M_k$ each provide $\rho$-zCDP. Then the mechanism $M(x) = (M_1(x \cap S_1), \ldots, M_k(x \cap S_k))$ provides $(z \cdot \rho)$-zCDP.
\end{lemma}

The proof of Lemma \ref{lem:generalized-parallel-composition-zcdp} requires the following property on the R\'enyi divergence, given in Lemma 2.2 of \cite{BunS16}.

\begin{lemma}(Lemma 2.2 of \cite{BunS16})
\label{lem:joint-divergence}

Let $P$ and $Q$ be distributions on $\Omega \times \Theta$. Let $P_\Omega$ and $Q_\Omega$ denote the marginal distributions on $\Omega$. Likewise, let $P_\Theta$ and $Q_\Theta$ denote the marginal distributions on $\Theta$. For $x \in \Omega$, let $P^x_\Theta$ and $Q^x_\Theta$ denote the conditional distributions on $\Theta$ conditioned on the first coordinate. Then
\begin{equation}
    D_{\alpha}(P \| Q) \leq D_{\alpha}(P_{\Omega} || Q_{\Omega}) + \max_{x\in \Omega}D_{\alpha}(P^x_{\Theta} || Q^x_{\Theta})
\end{equation}

When $P$ and $Q$ are product distributions, then this becomes the following.
\begin{equation}
    D_{\alpha}(P \| Q) = D_{\alpha}(P_{\Omega} || Q_{\Omega}) + D_{\alpha}(P_{\Theta} || Q_{\Theta})
\end{equation}

\end{lemma}

With this, we can prove Lemma~\ref{lem:generalized-parallel-composition-zcdp}.

\begin{proof}[Proof of Lemma~\ref{lem:generalized-parallel-composition-zcdp}]
Suppose $x$ and $x'$ are neighbors and let $r$ be the (only) record in their symmetric difference.
Let $i_1, \ldots, i_j$ be the indices of the sets in $F$ containing $r$.
$j \le z$ since the maximum degree of $F$ is $z$.

\begin{align}
D_{\alpha}(M(x) \| M(x')) &= \sum_{i=1}^k D_{\alpha}(M_i(x \cap S_i) \| M_i(x' \cap S_i)) \\
&= \sum_{i \in \{i_1,\ldots,i_j\}} D_{\alpha}(M_i(x \cap S_i) \| M_i(x' \cap S_i)) \\
&\le \sum_{i \in \{i_1,\ldots,i_j\}} \alpha \cdot \rho \\
&\le \alpha \cdot (z \cdot \rho).
\end{align}
\end{proof}

\subsubsection{Postprocessing}
\label{sec:post-processing-background}

Zero-concentrated differential privacy is closed under postprocessing, meaning that the privacy guarantee cannot be weakened by manipulating the outputs of a zCDP mechanism without reference to the protected inputs.

\begin{lemma}(Postprocessing for zCDP \cite{BunS16})
\label{lem:post-processing}
Let $M: \mathcal{X} \rightarrow \mathcal{Y}$ and $f:\mathcal{Y} \rightarrow \mathcal{Z}$ be randomized algorithms. Suppose $M$ satisfies $\rho$-zCDP. Then $f \circ M: \mathcal{X} \rightarrow \mathcal{Z}$ satisfies $\rho$-zCDP.
\end{lemma}

\subsubsection{Base Mechanism}
\label{sec:base-mechansisms}

\begin{definition}[L2 Sensitivity]
\label{def:sensitivity}
Given a vector function $q: \mathcal{X} \rightarrow \mathbb{Z}^n$, the sensitivity of $q$ is \linebreak $\sup_{x, x'} \|q(x)-q(x')\|_2$ where $x$ and $x'$ are neighboring databases and $\| \cdot \|_2$ is the Euclidean norm.
\end{definition}

There is an equivalent notion of bounded sensitivity.
\begin{definition}[Bounded L2 Sensitivity]
\label{def:bounded_sensitivity}
Given a vector function $q: \mathcal{X} \rightarrow \mathbb{Z}^n$, the sensitivity of $q$ is $\sup_{x ,x'} \|q(x)-q(x')\|_2$ where $x$ and $x'$ are bounded-neighboring databases and $\| \cdot \|_2$ is the Euclidean norm.
\end{definition}

\begin{definition}
\label{def:discrete-gaussian-distribution}
The discrete Gaussian distribution $\mathcal{N}_{\mathbb{Z}}(\sigma^2)$ centered at 0 is
\begin{equation}
\forall x \in \mathbb{Z}, \quad \Pr[X=x] = \frac{e^{-x^2/2\sigma^2}}{\sum_{y \in \mathbb{Z}}e^{-y^2/2 \sigma^2}}.
\end{equation}
\end{definition}

\begin{lemma}{\cite{CanonneK2020}}
\label{lem:discrete-gaussian-satisfies-zcdp}
Let $q: \mathcal{X} \rightarrow \mathbb{Z}^n$ with L2 sensitivity $\Delta$. 
Then outputting \textsc{NoisyCount}$(q(x), \rho)$ from Algorithm~\ref{alg:base-discrete-gaussian} satisfies $\Delta^2\rho$-zCDP. 
\end{lemma}

\begin{lemma}{\cite{CanonneK2020}}
\label{lem:discrete-gaussian-satisfies-bounded-zcdp}
Let $q: \mathcal{X} \rightarrow \mathbb{Z}^n$ with bounded L2 sensitivity $\Delta$. 
Then outputting \textsc{NoisyCount}$(q(x), \rho)$ from Algorithm~\ref{alg:base-discrete-gaussian} satisfies bounded $\Delta^2\rho$-zCDP. 
\end{lemma}

%%% Local Variables:
%%% mode: latex
%%% TeX-master: "main"
%%% End:

\subsection{Privacy Analysis}

With the provided background, we are ready to prove that SafeTab-P satisfies zCDP.

\begin{theorem}
\label{thm:safetab-satisfies-zcdp}
Let $\rho_{total} = \sum_{i=1}^{\omega} \rho_i$. Algorithm~\ref{alg:safetab-main-algorithm} satisfies (unbounded) $\rho_{total}$-zCDP.
\end{theorem}
\begin{proof}
The proof of Theorem~\ref{thm:safetab-satisfies-zcdp} follows via the combination of composition rules along with the fact that the \textsc{NoisyCount} procedure satisfies zCDP, per Lemma~\ref{lem:discrete-gaussian-satisfies-zcdp}, with respect to its inputs. In other words, we can think of \textsc{NoisyCount} as a basic building block from which we construct SafeTab-P.  

First, we construct the \textproc{TabulatePopulationGroup} procedure. That is, we claim that the procedure \textproc{TabulatePopulationGroup} in Algorithm~\ref{alg:safetab-tabulate-pop-group} satisfies $\rho$-zCDP with respect to the input dataframe, where $\rho$ is the privacy parameter input to the procedure.
Note that \textproc{TabulatePopulationGroup} actually uses one of two algorithms, depending on whether the population group is in the set TotalOnly. We consider each of these algorithms.

\textbf{Case 1:} $P \in $ TotalOnly. In this case, the procedure simply calls \textsc{NoisyCount}, which has L2 sensitivity 1, and satisfies $\rho$-zCDP.

\textbf{Case 2:} $P \not \in $ TotalOnly. In this case, the procedure can be decomposed into two parts. First, we call \textsc{NoisyCount}, which has L2 sensitivity 1,  with a budget of $\epsfrac\rho$. Then, we use the result to group the data by sex and age, we make a call to \textsc{NoisyCount}, over the vector of all the sex by age categories, with a budget of $(1-\epsfrac)\rho$.
The composition of the calls on all the sex by age groups satisfies $(1-\epsfrac)\rho$ by Lemma~\ref{lem:generalized-parallel-composition-zcdp}.
The (adaptive) composition of the two parts has total privacy loss $\rho$ by Lemma~\ref{lem:sequential-composition-zcdp}.

Next, we claim that the $i$th loop of the \textbf{for} loop on line~\ref{line:pop-group-level-loop} of Algorithm~\ref{alg:safetab-main-algorithm} satisfies $\rho_i$-zCDP.
By the definition of $s$, any particular record can appear in the input ($df_P$) of at most $s$ calls to \textsc{TabulatePopulationGroup}.
Therefore, by Lemma~\ref{lem:generalized-parallel-composition-zcdp}, the total privacy loss of the loop is $s$ times the privacy loss of \textsc{TabulatePopulationGroup}, i.e. $s \cdot \frac{\rho_i}{s} = \rho_i$.

Finally, the overall algorithm satisfies $(\sum_{i=1}^{\omega} \rho_i)$-zCDP by Lemma~\ref{lem:sequential-composition-zcdp}.
\end{proof}

\begin{theorem}
\label{cor:safetab-satisfies-bounded-zcdp}
Let $\rho_{total} = \sum_{i=1}^{\omega} \rho_i$. Algorithm~\ref{alg:safetab-main-algorithm} satisfies bounded $2\rho_{total}$-zCDP.
\end{theorem}

\begin{proof}

The proof of Theorem~\ref{cor:safetab-satisfies-bounded-zcdp} follows by tracking privacy loss across the joint mechanisms along with the fact that the \textsc{NoisyCount} procedure satisfies bounded zCDP, per Lemma~\ref{lem:discrete-gaussian-satisfies-bounded-zcdp}, with respect to its inputs. In other words, we can think of \textsc{NoisyCount} as a basic building block from which we construct SafeTab-P.  Under bounded zCDP, it is sufficient to bound the changes due to the removal of one record \textit{and} addition of another. As a result, we reason specifically about the effects of adding one record and removing another record.

First, we construct the \textproc{TabulatePopulationGroup} procedure. That is, we claim that the procedure \textproc{TabulatePopulationGroup} in Algorithm~\ref{alg:safetab-tabulate-pop-group} satisfies bounded $2\rho$-zCDP with respect to the input dataframe, where $\rho$ is the privacy parameter input to the procedure.

Note that \textproc{TabulatePopulationGroup} actually uses one of two algorithms, depending on whether the population group is in the set TotalOnly. Likewise, we also need to consider if both the added and deleted record belong to the same population group or different population groups. For each of these cases, we bound the R\'enyi divergence of the mechanism's output distribution for bounded neighboring databases and demonstrate under which conditions it is maximized.

\textbf{Case 1:} Assume $P \in $ TotalOnly and both records are in $P$. In this case, the procedure simply calls \textsc{NoisyCount}. Since a record is both added and removed from this count, the distribution of this count does not change and as such the R\'enyi divergence is $0$.

\textbf{Case 2:} Assume $P \in $ TotalOnly and either the added record or removed record is in $P$ but not both. In this case, the procedure simply calls \textsc{NoisyCount}. This count changes by $1$ and, by Lemma~\ref{lem:discrete-gaussian-satisfies-bounded-zcdp}, has a maximum R\'enyi divergence of $\alpha \rho$.

\textbf{Case 3:} Assume $P \not \in $ TotalOnly and both records are in $P$. In this case, the procedure can be decomposed into two parts. First, we call \textsc{NoisyCount} with a budget of $\epsfrac\rho$. Then, we use the result to do a GroupBy on the data by sex and age, and we make a call to \textsc{NoisyCount} with a budget of $(1-\epsfrac)\rho$ to compute the counts for all sex by age categories.
The first output of \textsc{NoisyCount} has a R\'enyi divergence of $0$ since the count does not change when both records are in the same population group. 
By Lemma~\ref{lem:joint-divergence}, the R\'enyi divergence of the combined outputs of \textsc{NoisyCount} is bounded by the R\'enyi divergence of the first output of \textsc{NoisyCount} plus the R\'enyi divergence of the second output of \textsc{NoisyCount} when conditioned on the result of the first call that maximizes the total R\'enyi divergence. This bound holds for any of the sex by age marginals. 
In this case, one category increases by $1$ and the other decreases by $1$. Therefore, the outputs of \textsc{NoisyCount} have a R\'enyi divergence of $\alpha 2(1-\epsfrac)\rho$ by Lemma~\ref{lem:discrete-gaussian-satisfies-bounded-zcdp}. Then, by  Lemma~\ref{lem:joint-divergence}, the (adaptive) composition of the two parts has a total divergence of $\alpha 2(1-\epsfrac)\rho$.

\textbf{Case 4:} Assume $P \not \in $ TotalOnly and either the added record or removed record is in $P$ but not both. In this case, the procedure can be decomposed into two parts. First, we call \textsc{NoisyCount} with a budget of $\epsfrac\rho$. Then, we use the result to do a GroupBy on the data by sex and age, and make a call to \textsc{NoisyCount}, over the vector of all sex by age categories, with a budget of $(1-\epsfrac)\rho$. The first output of \textsc{NoisyCount} has a R\'enyi divergence of  $\alpha \epsfrac\rho$ by Lemma~\ref{lem:discrete-gaussian-satisfies-bounded-zcdp} since the count changes by $1$ when only one of the records is in the population group. 
By Lemma~\ref{lem:joint-divergence}, the R\'enyi divergence of the combined outputs of \textsc{NoisyCount} is bounded by the R\'enyi divergence of the first output of \textsc{NoisyCount} plus the R\'enyi divergence of the second output of \textsc{NoisyCount} when conditioned on result of the first call that maximizes the total R\'enyi divergence. 
In this case, one sex by age category either increases by $1$ or decreases by $1$. Therefore, the outputs of \textsc{NoisyCount} have a R\'enyi divergence of $\alpha (1-\epsfrac)\rho$ by Lemma~\ref{lem:discrete-gaussian-satisfies-bounded-zcdp}. Then, by Lemma~\ref{lem:joint-divergence}, the (adaptive) composition of the two parts has total R\'enyi divergence of $\alpha\rho$.

Next, we claim that the $i$th loop of the \textbf{for} loop on line~\ref{line:pop-group-level-loop} of Algorithm~\ref{alg:safetab-main-algorithm} satisfies $2\rho_i$-zCDP. By the definition of $s$, any particular record can appear in the input ($df_P$) of at most $s$ calls to \textsc{TabulatePopulationGroup}. For each of the $s$ calls for the added record, the removed record can either overlap in their population group or be part of a different population group.
When they overlap (Cases 1,3), the maximum R\'enyi divergence is of $\alpha 2(1-\epsfrac)\rho_i/s$. This occurs when the population group is not in TotalOnly (Case 3). When they do not overlap (Cases 2,4), the R\'enyi divergence is the same for population groups in TotalOnly and not in TotalOnly. For these, each output of  \textsc{TabulatePopulationGroup} has a R\'enyi divergence of $\alpha \rho_i/s$. However, since both the added and removed records are part of different population groups, the output of \textsc{TabulatePopulationGroup} for each of those population groups has a R\'enyi divergence of $\alpha \rho_i/s$ resulting in a combined R\'enyi divergence of $\alpha 2 \rho_i/s$ by Lemma~\ref{lem:joint-divergence}. Since $\alpha 2 \rho_i/s > \alpha 2(1-\epsfrac)\rho_i/s$, the maximum R\'enyi divergence occurs when none of the population group for either the added or removed record overlap.
Since this can happen at most $s$ times the total R\'enyi divergence is $s(\alpha 2 \rho_i/s)=\alpha 2 \rho_i$ by Lemma~\ref{lem:joint-divergence}.
This is the case that maximizes the R\'enyi divergence over the entire \textbf{for} loop and, therefore it satisfies bounded $2\rho_i$-zCDP.

Finally, the overall algorithm satisfies $(\sum_{i=1}^{\omega} 2\rho_i)$-zCDP by Lemma~\ref{lem:sequential-composition-zcdp}.

\end{proof}

\section{Implementation of SafeTab-P}
\label{sec:implementation}
The algorithm presented in Section \ref{sec:algorithm-description} is a simplified version of the implemented SafeTab-P program. In this section, we describe some of the differences between the implementation and the simplified algorithm. We focus on differences that could affect the privacy calculus and describe why the implementation is equivalent to the simplified algorithm.

\subsection{Input Validation}\label{sec:validation}
Input validation is an important step before deploying a differentially private algorithm. SafeTab-P performs extensive validation of its input data to ensure that provided data are in the expected formats and internally consistent. Some validation occurs on data that are considered public (like the list of all geographic entities for which data is to be tabulated) and therefore, do not affect the privacy guarantees. However, we also validate the private input files that contain the data of individual census respondents. 

Validation failures are only visible to the trusted curator running the program, and on failure, no part of the differentially private program is run. Validation failures are made available to the trusted curator, so they can correct any errors in the provided input files before executing the differentially private program. Failures in validation are not released publicly and therefore, do not contribute to any privacy loss.

\subsection{Tumult Analytics}\label{sec:analytics}
Rather than directly calculating stability and sampling from noise distributions, SafeTab-P is implemented using Tumult Analytics\cite{berghel2022tumult}, a framework for implementing differentially private queries. 

A key benefit of using Tumult Analytics is that all access to the sensitive data is mediated through a Tumult Analytics \texttt{session}. The \texttt{session} tracks all the transformations and measurements performed on the sensitive data and is able to correctly compute the total privacy loss of the computation on the sensitive data. In SafeTab-P, we construct an Analytics \texttt{session} with:
\begin{itemize}
    \item The total privacy-loss budget for the pipeline (calculated as the sum of all population-group budgets $\rho_{i}$).
    \item The private dataset.
    \item The public datasets (information on all race and ethnicity codes, characteristic iterations, and geographic entities).
    \item The neighboring definition (privacy is with respect to the addition/removal of one record from the private dataset).
    \item The privacy definition to be satisfied (e.g., zCDP).
\end{itemize}
We then implement all data transformations (like mapping a record to its characteristic iterations) and queries within the framework. Tumult Analytics tracks the stability throughout the transformations and applies an appropriate amount of noise to the final queries, guaranteeing that the outputs are differentially private and no more than the total budget is expended.

The use of Tumult Analytics means that, while the simplified algorithm above only describes $\rho$-zCDP, SafeTab-P also supports "pure" $\epsilon$-Differential Privacy with noise drawn from a double-sided geometric distribution. The user-provided privacy-loss budget values are interpreted as either $\epsilon$'s or $\rho$'s depending on which privacy definition the user selects in SafeTab-P's configuration. The production of Detailed DHC-A utilizes zCDP. Thus, the presentation of the algorithm in this document focus on zCDP.

The use of Tumult Analytics also allows (and necessitates) some other deviations from the simplified algorithm. Rather than counting each population group sequentially, in a for-loop, we use Analytics' \texttt{groupby} feature to tabulate many population groups at once. Analytics requires users of the \texttt{groupby} feature to specify all groups to tabulate in advance in the form of a \texttt{KeySet} object. In the case of SafeTab-P, we construct \texttt{KeySets} containing the combinations of possible geographic areas, characteristic iterations, and (if applicable) sex values and age buckets. We build these \texttt{KeySets} using separate input specification files (rather than relying on observed groups present in the CEF). \texttt{KeySets} do not use the confidential CEF data to ensure that noisy statistics are produced for every valid population group. Thus, we do not reveal whether population groups are empty via their presence or absence in the output data. We note the simplified algorithm also allowed for the possibility that $df_P$ (filtered dataframe that only contain records associated with population group $P$) is empty.

A key innovation in SafeTab-P is that statistics are released adaptively based on the total count of each population group. This is a challenging feature to implement in a manner that provably ensures differential privacy -- the privacy accounting across population groups that only output a Total count versus groups that output sex by age counts does not follow from adding up the privacy losses for each group (as in simple composition) but follows from a more nuanced generalized parallel composition (see Lemma~\ref{lem:generalized-parallel-composition-zcdp}). To ensure that we receive this tighter accounting, we use the Tumult Analytics \texttt{partition\_and\_create} operator to make the structure of the computation legible to the automatic privacy accountant.

\subsection{Postprocessing for Addressing Demographic Reasonableness Concerns}\label{sec:postprocessing}
After the differentially private algorithm has completed, we perform several additional postprocessing steps. Because these steps are purely postprocessing, they cannot affect the differential privacy guarantees per Lemma \ref{lem:post-processing}. These postprocessing steps are designed to address specific data quality consistency concerns that arise when adding noise to tabular statistics. All postprocessing was implemented under the direction of subject-matter experts at the Census Bureau. These steps do not exhaustively address all possible demographic reasonableness concerns. 

\subsubsection{Marginals}
A population group that is not in the $TotalOnly$ set can receive either a total count or a sex by age breakdown of varying granularity. For those that receive a sex by age breakdown, we calculate internally-consistent sex marginals (counts broken down by sex, but not age) and a total by summing the relevant noisy counts. This approach provides consistency within a sex by age table for a given population group but does nothing to provide consistency across different population groups.

\subsubsection{Suppression}\label{sec:output-suppression}
One demographic reasonableness concern associated with generating noisy counts is that of creating positive population group counts for groups that do not exist in the true data. We suppress population groups that received small noisy counts to combat this concern. The goal of the suppression is to reduce the probability that a true zero count is released as a positive noisy count. Beyond accomplishing that goal, suppression has two other consequences:
\begin{itemize}
  \item Groups with a true count larger than zero may be suppressed, including true counts larger than the suppression threshold.
  \item Suppression introduces a positive bias in released counts, especially for groups whose true count is at or below the suppression threshold. The bias varies depending on the magnitude of the true count relative to the suppression threshold with smaller counts exhibiting greater amounts of bias, as shown in Figure \ref{fig:bias}.
\end{itemize}

When analyzing suppression, we consider only the $\rho$-zCDP version of SafeTab-P as described in Section \ref{sec:algorithm-description} (which uses discrete Gaussian noise).

\paragraph{Suppression algorithm description.}

\begin{itemize}
  \item Run SafeTab-P to produce noisy counts.
  \item For every sub-state population group receiving only a total count, filter out those population groups whose noisy count is less than the suppression threshold $T$ (we consider a single threshold here for simplicity, but there may be different thresholds for different classes of population groups such as separate thresholds for detailed and regional population groups).
  \item Release the remaining noisy counts.
\end{itemize}

We only perform suppression on sub-state population groups that do not receive a sex by age breakdown since eligibility for sex by age requires passing a population threshold in the main SafeTab-P algorithm that is presumably higher than the postprocessing suppression threshold. Hence, population groups that do receive sex by age breakdowns are unlikely to have a true count of zero.

\paragraph{Computing the probability of zero count suppression.}

Rather than directly choosing a suppression threshold, $T$, we allow the SafeTab-P user to select a probability that a zero count is suppressed, $p$.
We then derive a threshold $T$ from $p$.
In the following, we let $invcdf(\sigma, x)$ denote the inverse cumulative distribution function (CDF) of the discrete Gaussian distribution \cite{CanonneK2020} with mean 0 and scale $\sigma$. As alluded to earlier, suppression is only applied at sub-state geography levels and thus, thresholds are only required for population groups that undergo the two stage adaptive process. The thresholds are set as 

\begin{equation}
  T = invcdf\left(\sqrt{\frac{9}{2(1-\gamma)\rho}}, p\right)
\end{equation}

\noindent where $\rho$ is the privacy-loss budget allocated to the population group level and $\gamma$ is the fraction of the budget reserved for the adaptive step and 9 is the stability value $\Delta(g_i)$ regardless of the index $i$.
In Table~\ref{tbl:target-thresholds}, we give threshold values for various values of $\rho$ for a 99.99\% target zero count suppression rate.

\begin{table}[h]
  \centering
  \begin{tabular}[h]{c c}
    \toprule
    Privacy loss ($\rho$) & Threshold (T) \\
    \midrule
    0.008 & 93 \\
    0.159 & 21 \\
    0.543 & 11  \\
    \bottomrule
  \end{tabular}
  \caption{Suppression thresholds for various values of $\rho$ to give a 99.99\% chance that a true zero count will be suppressed. We set $\gamma = 0.1$ in all cases.}\label{tbl:target-thresholds}
\end{table}

\paragraph{Computing the probability of nonzero count suppression.}

One side effect of suppressing true zeros (i.e., groups whose true count is zero) with high probability is that small population groups will also be suppressed.
In Figure~\ref{fig:suppression-probability}, we plot the probability that a population group will be suppressed, given that true zeros are suppressed with 99.99\% probability.
We plot the true count of the population group as a fraction of the suppression threshold $T$ and therefore, this plot gives the correct probabilities regardless of the privacy loss (as long as $T$ represents the threshold for which zero counts are suppressed with probability 99.99\%.)

\begin{figure}[H]
  \centering
  \includegraphics[width=0.7\textwidth]{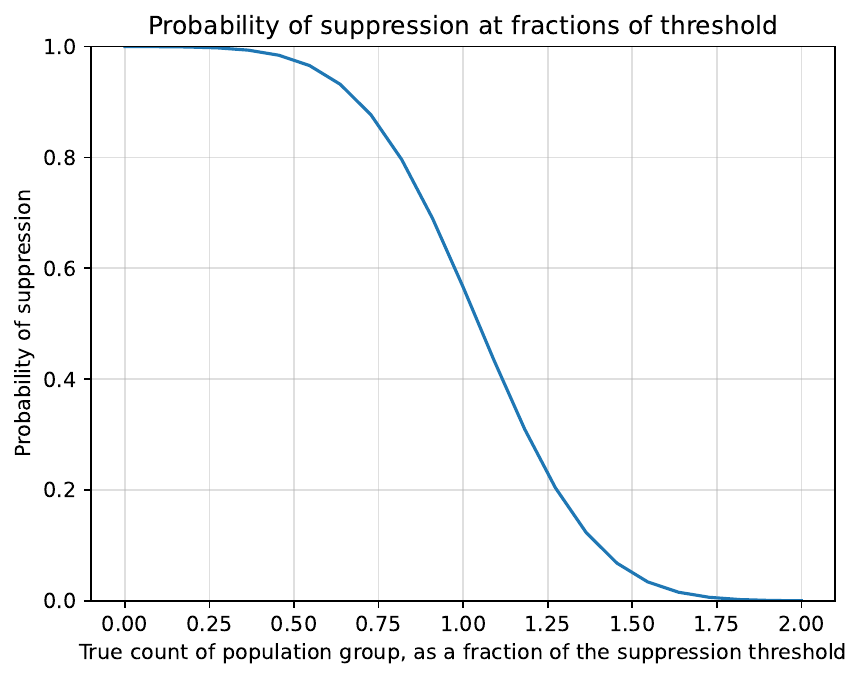}
  \caption{The probability of suppression as a function of the true size of the population group, when true zeros are suppressed with probability 99.99\%. The true size of the population group is expressed as a fraction of the suppression threshold. For example, the ``0.5'' indicates a population group with true size $T/2$, where $T$ is the suppression threshold.}\label{fig:suppression-probability}
\end{figure}

\paragraph{Bias in released counts.}

Another side effect of suppression is that released counts have a positive bias.
The cause is easy to understand for population groups whose true count is below the suppression threshold: noisy estimates for these counts will only be released if the algorithm adds sufficient positive noise.
However, noisy estimates for population groups whose true count is larger than the suppression threshold also have a positive bias.
These groups are suppressed if we add sufficient negative noise.

We can calculate the bias by calculating the expected noise added to a population group, given the population group is published.
If we let $n$ be the true size of the population group, we can use the following formula to calculate the expected noise:

\begin{equation}
  \label{eq:2}
  \mathbb{E}[X | n + X > T] = \frac{\sigma \cdot \phi(\frac{T-n}{\sigma})}{1 - \Phi(\frac{T - n}{\sigma})},
\end{equation}

\noindent where $\phi(\cdot)$ is the probability density function of the standard discrete Gaussian distribution, $\Phi(\cdot)$ is the cumulative distribution function of the standard discrete Gaussian distribution, $X$ is a random variable representing the amount of noise added to the population group count, and $\sigma$ is the scale of the noise.

In Figure~\ref{fig:bias}, we plot the expected noise added to a population group count, given the population group is released.
As in the last section, we give the true size of the population group and the bias as fractions of the suppression threshold.
This means the results hold regardless of the scale of the noise, as long as the suppression threshold, $T$, represents the threshold for which zero counts are suppressed with probability 99.99\%.

\begin{figure}[H]
  \centering
  \includegraphics[width=0.7\textwidth]{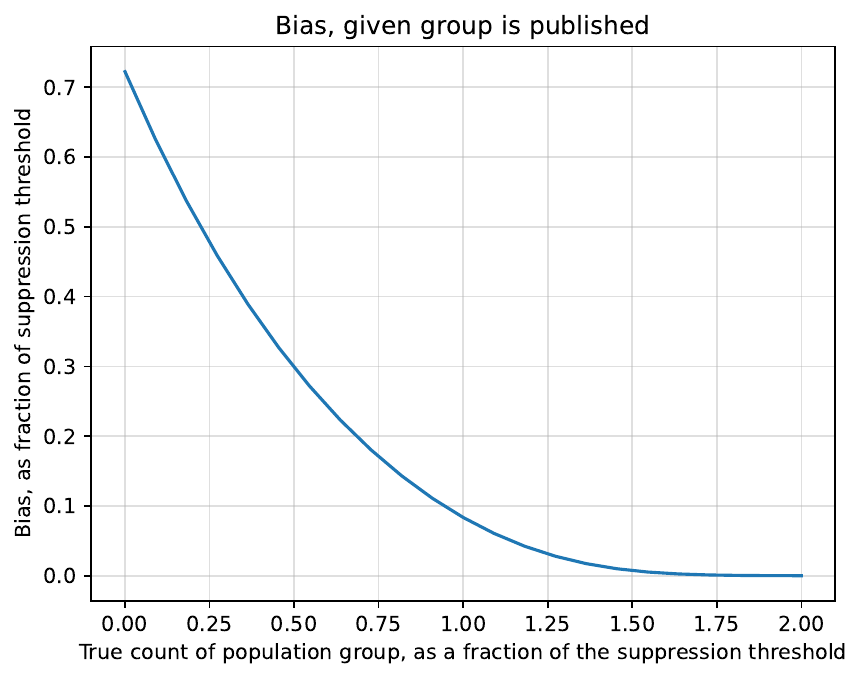}
  \caption{The bias in reported population group size compared to the true size of the population group, when true zeros are suppressed with probability 99.99\%. Both the size of the population group and the bias are expressed as fractions of the suppression threshold. For example, the ``0.5'' on the x-axis indicates a population group with size $T/2$, where $T$ is the suppression threshold and ``0.5'' on the y-axis indicates a bias of $T/2$.}\label{fig:bias}
\end{figure}

\subsubsection{Coterminous Geographies}
Sometimes, two or more geographic entities in different geographic summary levels share the same geographic boundaries (i.e., they are aggregated from identical collections of Census blocks). For example, Washington D.C. is tabulated as a state, county, and place. These geographic entities are called \emph{coterminous}. Another coterminous example is a county containing a single Census tract. A characteristic iteration receives different independent noisy measurements for each geographic summary level of a given coterminous area. However, its counts should be identical at each summary level. It is also possible that a characteristic iteration is suppressed at one summary level of a coterminous area but not suppressed in the other summary levels. Some geographic entities at different summary levels that do not share the same geographic boundaries should still be statistically equivalent. For example, if a county contains one water-only tract and one nonwater tract, characteristic iterations should have the same counts in the nonwater tract as in the county. We abuse terminology by calling statistically equivalent geographic areas coterminous. Working in consultation with subject-matter experts at the Census Bureau, we created a postprocessing step (implemented as a standalone program) that corrects inconsistencies in coterminous geographic entities. We assume as input a hierarchy over geographic summary levels. For each set of coterminous population groups (e.g., \{(D.C. (as a State), Maori Alone), (D.C. (as a County), Maori Alone), (D.C. (as a Place), Maori Alone)\}), we find the population group with the hierarchically highest geographic summary level that has not been suppressed, called the \emph{donor} population group. We then overwrite the data of the other population groups in the set with the data of the donor population group. A pseudocode description of this procedure is presented in Algorithm \ref{alg:coterminous}.

\subsubsection{Tabulation System Suppression}
Additional demographic reasonableness corrections are addressed outside the DAS. In particular, the Decennial Tabulation System also performs suppression postprocessing. Again, per Lemma \ref{lem:post-processing}, this does not impact the privacy analysis. The suppression conducted by non-DAS systems is out of scope for this paper but includes further enforcement of nonnegativity as well as suppression in cases where noisy counts for Alone characteristic iterations appear to be greater than the corresponding noisy counts for Alone or in Any Combination characteristic iterations.

\begin{algorithm}[t]
\caption{\label{alg:coterminous} The coterminous geographies postprocessing algorithm.}
\begin{algorithmic}[1]
\Require $df$: The output of the main SafeTab-P algorithm
\Require $\mathcal{C}$: A list of all sets of coterminous population groups.
\Require $R$: An ordering of geographic summary levels.

\Procedure{Coterminous geographies}{$df$, $\mathcal{C}$, $R$}

\For{$C \in \mathcal{C}$}
\For{$r \in R$}
\State {$P \leftarrow \text{The population group in } C \text{ whose geographic entity is at summary level } r$.}
\If{$df$ contains data for $P$}
\State \text{Let $M$ be the set of counts in $df$ for $P$}
\For{$P' \in C - \{P\}$}
\State \text{Remove all counts for $P'$ from $df$}
\State \text{Add $M$ to $df$ for $P'$.}
\EndFor
\EndIf
\EndFor
\EndFor
\EndProcedure
\end{algorithmic}
\end{algorithm}

\subsection{Other Implementation Details}\label{sec:complication}
We note a few other implementation differences that are primarily driven by the specification requirements of the system and data wrangling aspects of the code. These include the function mapping race/ethnicity codes to characteristic iterations, rules for what statistics are tabulated for different population groups, and the handling of Puerto Rico. 

\subsubsection{Mapping Race/Ethnicity Codes to Characteristic Iterations}
The pseudocode in Section \ref{sec:algorithm-description} abstracts the process of mapping a person's input record into their corresponding population group as a function $g_{i}$. In practice, this process requires joining against several specification input files and some subtle logic (to determine whether a user qualifies for an Alone characteristic iteration in addition to an Alone or in Any Combination characteristic iteration). However, the end result is functionally equivalent to the $g_{i}$ abstraction - each record is mapped to a number of geographic entities and characteristic iterations. The stability factor of the implemented transformations, the equivalent of $\Delta(g_i)$, is automatically tracked by Analytics rather than being computed by hand. 

The pseudocode also ignores the logic associated with pre-processing a specified universe of geographic entities and iteration codes into population group levels $\mathcal{P}_i$. That is, the master list of all population groups divided into population group levels is constructed through a combination of specification files rather than being handed directly to the system.

\subsubsection{Puerto Rico}\label{sec:pr}
The SafeTab-P algorithm presented in Algorithm \ref{alg:safetab-main-algorithm} describes an input dataframe consisting of records of every person in the United States. However, the same algorithm is applied to a dataframe consisting of records from Puerto Rico. In implementation, SafeTab-P tabulates data for the United States and Puerto Rico in two separate passes.

\section{Parameters and Tuning}\label{sec:params}

Between the pseudocode representation of SafeTab-P and the selected implementation details presented earlier, we have alluded to a number of parameters that must be set before executing a run of the SafeTab-P program. Parameters are adjustable factors that must be fixed to fully define the nature of the program (e.g., the noise distribution employed in \textsc{NoisyCount}, the privacy-loss budgets for population group levels, and population thresholds). With regard to SafeTab-P specifically (but any differentially private algorithm generally), parameter selection is a matter of policy. The Census Bureau's Data Stewardship Executive Policy (DSEP) committee, in consultation with subject-matter experts as well as internal and external privacy experts, approved all available parameters for the Detailed DHC-A. Parameter selection necessitates trade-offs, as many of these parameters are dependent on each other. To illustrate, we consider a fundamental relationship between the privacy loss parameters and their corresponding margins of error with discrete Gaussian noise distributions.

\subsection{Error bounds}\label{sec:discrete-gauss-error-bound}
We derive error bounds for Algorithm \ref{alg:safetab-main-algorithm} with discrete Gaussian noise. We begin by stating a portion of Proposition 25 from \cite{CanonneK2020}.

\begin{proposition}[Proposition 25 in \cite{CanonneK2020}]\label{lem:discrete-gaussian-bound}

For all $m \in \mathbb{Z}$ with $m \geq 1$ and for all $\sigma \in \mathbb{R}$ with $\sigma > 0$, $\Pr[X \geq m]_{X \leftarrow \mathcal{N}_{\mathbb{Z}}(\sigma^2)} \leq \Pr[X \geq m-1]_{X \leftarrow \mathcal{N}(\sigma^2)}$.
\end{proposition}

The following corollary is immediate. 
\begin{corollary}
For all $x, \sigma \in \mathbb{R}$ with $x \geq 1$ and $\sigma > 0$,  $\Pr[X > x]_{X \leftarrow \mathcal{N}_{\mathbb{Z}}(\sigma^2)} \leq \Pr[X > \lfloor x \rfloor]_{X \leftarrow \mathcal{N}(\sigma^2)}$.
\end{corollary}

Figure 2 of \cite{CanonneK2020} provides an intuitive visualization of these tail bounds.
It follows that $X \in [- \lfloor 1.96\sigma \rfloor,  \lfloor 1.96\sigma \rfloor]$ with probability at least 95\%.
That is, the 95\% margin of error (MOE), defined as half the width of the 95\% confidence interval, is given by $\lfloor 1.96\sigma \rfloor$. 

Hence, for a population group in level $i$ in the TotalOnly set, the MOE in the directly computed total estimate from line \ref{line:noisy-total-only} in Algorithm \ref{alg:safetab-tabulate-pop-group}  is $\left \lfloor1.96\sqrt{\frac{s}{2\rho_i}} \right \rfloor$ ($s = \Delta(g_i) = 9$ and $\rho_i$ is the privacy-loss budget for level $i$). For the population groups in level $i$ not in the TotalOnly set, the MOE in a single sex by age group in Algorithm \ref{alg:safetab-tabulate-pop-group} is 
$\left \lfloor1.96\sqrt{\frac{s}{2(1-\gamma)\rho_i}} \right \rfloor$ ($s = \Delta(g_i) = 9$, $\rho_i$ is the privacy-loss budget for level $i$, and $\gamma$ is the fraction of the budget used in Step 1 of the adaptive process).

\begin{corollary}\label{cor:dgauss-rho-from-moe}
The NoisyCount procedure implemented with discrete Gaussian noise and run with $\rho = \frac{1.92}{\lfloor MOE \rfloor^2}$ has a 95\% MOE of at most $MOE$.
\end{corollary}

This relationship dictates that adjusting a desired MOE to achieve improved data utility comes at the cost of additional privacy loss. Likewise, adjusting privacy-loss parameters to provide better privacy protection comes at the cost of wider MOE in the estimated statistics.

\subsection{Parameter Identification, Trade-offs, and Outcomes}

Fully identifying the set of possible parameters, let alone navigating the trade-offs associated with them, was no small feat. Before delving into specific parameters, we will overview some of the methods involved in navigating associated trade-offs. Firstly, parameters tend to impact some combination of these three aspects: data accessibility, data accuracy, and privacy. Data accessibility, in a nutshell, refers to the volume of tabular statistics released. Data accuracy is primarily measured by margins of error on the tabular statistics. Privacy is measured by the privacy-loss parameters of the algorithm. For example, excluding population group level $i$ would reduce data accessibility but improve privacy since the privacy loss $\rho_i$ is no longer necessary. To aid in the understanding of these trade-offs, we relied on a combination of tangible tools and theoretical analyses (such as the suppression analysis demonstrated in Section \ref{sec:output-suppression}).  

Before development on SafeTab-P began, we created the SafeTEx tool to provide hands-on experience with some of the trade-offs. We will provide a brief summary of this tool. Then, we will highlight specific parameters and the critical decisions made by the DSEP committee as a consequence of either interacting with the tool, reviewing theoretical analysis, or subject-matter expert guidance.

\subsection{Parameter Tuning using the SafeTEx Tool}
SafeTEx is an easy-to-use interactive decision support tool implemented in the Microsoft Excel program and developed to facilitate conversations between subject-matter experts, disclosure avoidance scientists, and ultimately, the DSEP committee. 

The tool allowed users to interactively specify: 
\begin{itemize}
\item The set of geography levels and characteristic iteration levels that constitute the universe of population groups for which statistics are tabulated. 
\item The maximum number of races a person is associated with to an integer in the range 1-8.
\item A target error threshold.
\item A privacy-loss budget $\rho$ and the fraction of the privacy-loss budget reserved for the Step 1 of the adaptive procedure $\gamma$.
\end{itemize}
Based on these parameters, the SafeTEx tools computed thresholds on the size of population groups at which different statistics levels (total, sex by age(4), sex by age(9), sex by age(23) can be released at the target error.  The computations were performed using analytical formulae for expected error of noise mechanisms employed in the SafeTab-P algorithm. 

We describe in the next section some of the key parameters considered and the decision process used to set these parameters. 

\subsubsection{Parameter Selection}

\noindent \emph{\textbf{Population group levels, TotalOnly population groups, and population thresholds}}: As a reminder, a population group level is defined by a geographic summary level, such as Nation, State, and County, and an iteration level (i.e., Detailed or Regional). The SafeTEx tool helped subject-matter experts grasp the impact of adding or removing levels would have on the privacy-loss budget, but they had to weigh that against the value of having publishable statistics at each given level. They also gathered feedback from the user community on these topics and ultimately settled on the levels referenced in Table \ref{tab:moe-targets}. Furthermore, SafeTEx helped subject-matter experts reason concerning the relation between absolute expect error and population sizes. These relative error comparisons were instrumental in determining population thresholds for the adaptive algorithm. We leave the research details of these decision-making efforts as a topic for a future paper. The subject-matter experts also made a determination as to which groups should be pre-defined as TotalOnly. The SafeTEx tooling did not directly address this concept. 

\vspace{\baselineskip}

\begin{table}[t]
    \centering
    \begin{tabular}{c c c c c c}
    \toprule
    Population Group Level & MOE Target &  
    \multicolumn{2}{c}{Unbounded Privacy Loss} & \multicolumn{2}{c}{Bounded Privacy Loss}
    \\
    \cmidrule(lr{.75em}){3-4}
    \cmidrule(lr{.75em}){5-6}
    & & Step 2 & Total & Step 2 & Total \\
    \midrule
    (Nation, Detailed): $\rho_1$ & 3 & 1.921 & 2.134 & 3.842 & 4.268\\
    (State, Detailed): $\rho_2$ & 3 & 1.921 & 2.134 & 3.842 & 4.268\\
    (County, Detailed): $\rho_3$ & 11 & 0.143 & 0.159 & 0.286 & 0.318 \\
    (Tract, Detailed): $\rho_4$ & 11 & 0.143 & 0.159 & 0.286 & 0.318 \\
    (Place, Detailed): $\rho_5$ & 11 & 0.143 & 0.159 & 0.286 & 0.318 \\
    (AIANNH, Detailed): $\rho_6$ & 11 & 0.143 & 0.159 & 0.286 & 0.318\\
    (Nation, Regional): $\rho_7$ & 50 & 0.007 & 0.008 & 0.014 & 0.016\\
    (State, Regional): $\rho_8$ & 50 & 0.007 & 0.008 & 0.014 &0.016\\
    (County, Regional): $\rho_9$ & 50 & 0.007 & 0.008 & 0.014 &0.016\\
    (Tract, Regional): $\rho_{10}$ & 50 & 0.007 & 0.008 & 0.014 &0.016\\
    (Place, Regional): $\rho_{11}$ & 50 & 0.007 & 0.008 & 0.014 &0.016\\
    \bottomrule
    \end{tabular}
   \caption{MOE targets for the statistics released (in Step 2 of the adaptive algorithm) at different population group levels, along with the corresponding privacy loss (unbounded and bounded $\rho$-zCDP for discrete Gaussian). The privacy loss is reported for the Step 2 (to match the MOE) as well as the total loss for that level. Step 2 loss is 90\% of Total loss at each population group level. Note that the privacy losses reported here have already been aggregated over all the population groups at the given level, so the \emph{Total} column represents the privacy loss input parameters of the SafeTab-P algorithm.}
   \label{tab:moe-targets}
\end{table}

\noindent \emph{\textbf{Noise distribution}}: As mentioned in the implementation details, SafeTab-P can either be instantiated with discrete Gaussian noise or geometric noise. Other noise distributions were not considered because of a desire to preserve integrality in the outputs without needing to postprocess the results to achieve it. We analytically compared the privacy and accuracy of two noise distributions and observed that the discrete Gaussian mechanism displayed roughly 7\% less privacy loss, compared to the geometric mechanism at the same accuracy targets. We analyzed both mechanisms under another alternative to pure differential privacy known as approximate differential privacy to ensure consistency in the assessment of privacy. Details of this analysis can be found in \cite{Haney2021}. As an outcome of this analysis, the Census Bureau decided to use discrete Gaussian noise. Hence, the focus of this paper on that mechanism. 

\vspace{\baselineskip}

\noindent \emph{\textbf{Race Multiplicity}}: The stability $\Delta(g_i)$ of the flatmap transformation $g_i$ mapping individuals to population groups in level $i$ is a significant factor in the noise scale required to satisfy zCDP. The data collection process restricts individuals to a maximum of 8 detailed race codes and 1 ethnicity code, which translates to a flatmap stability of 9 for any given population group level. This is because an individual with 8 unique race codes can be associated with at most 8 Alone or in Any Combination characteristic iterations for a level plus the one ethnic characteristic iteration. Higher stability equates to higher variance noise, all else held equal, so an option to improve the noise variance would be to reduce the stability by setting a lower cap on the number of race codes processed for each individual. For example, if individuals were restricted to 3 race codes instead of 8, the stability would drop from 9 to 4 resulting in a 33\% decrease in MOE when holding the privacy loss constant. However, the restriction would also introduce another form of bias into the statistics and potential artificially reduce the set of population groups with true positive counts. The DSEP committee opted to leave the race multiplicity parameter at 8.

\vspace{\baselineskip}

\noindent \emph{\textbf{MOE, $\rho$, and $\gamma$}}: Recall that $\gamma$ is the fraction of $\rho$ reserved for the Step 1 of the adaptive procedure in SafeTab-P. The SafeTEx tool provided an interface for adjusting $\rho$'s and $\gamma$ to observe the impact on expected MOE as derived in Section \ref{sec:discrete-gauss-error-bound}. The Census Bureau selected $\gamma = 0.1$ and  set $\rho_i$'s as displayed in Table \ref{tab:moe-targets} for the production run of SafeTab-P on the 2020 Census data.

\section{Conclusion}\label{sec:conclusion}

In this article, we presented SafeTab-P, a differentially private algorithm designed for releasing the Detailed DHC-A data product. We explained the adaptive nature of SafeTab-P and proved the algorithm satisfies zCDP. We also looked at selected implementation details, including how the SafeTab-P program was built on the Tumult Analytics platform. We described our contributions to the tuning of parameters for SafeTab-P. In future papers, we will discuss algorithms designed for the release of other 2020 Census data products, such as the Detailed DHC-B and S-DHC.

\bibliographystyle{plain}
\bibliography{refs}

\end{document}